\documentclass[aps,pra,superscriptaddress,tightenlines,longbibliography,,amsmath,amssymb,twocolumn]{revtex4-1} 
\usepackage{graphicx,natbib,physics,fontenc,amsthm,textcomp,color}
\newtheorem{theorem}{Theorem}
\theoremstyle{definition}
\newtheorem{defn}{Definition}
\newtheorem{ftr}{Feature}
\newcommand{\saa}{\mathfrak{S}_{A\rightarrow A}}
\newcommand{\saab}{\mathfrak{S}_{A\rightarrow AB}}
\newcommand{\half}{\frac{1}{2}}
\begin{document}
\title{The operational reality of quantum nonlocality}
\author{R. Srikanth}
\affiliation{Theoretical Sciences Division,
	Poornaprajna Institute of Scientific Research (PPISR), \\
	Bidalur post, Devanahalli, Bengaluru 562164, India}

\begin{abstract}
Does the remote measurement-disturbance of the quantum state of a system $B$ by measurement on system $A$ entangled with $B$, constitute a real disturbance -- i.e., an objective alteration--  of $B$ in an operational sense?  Employing information theoretic criteria motivated by operational considerations alone, we argue that the disturbance in question is  real for a subset of steerable correlations. This result highlights the distinction between quantum no-signaling and relativistic signal-locality. It furthermore suggests a natural reason why a convex operational theory should be non-signaling: namely, to ensure the consistency between the properties of reduced systems and those of single systems.
\end{abstract}
\maketitle

\section{Introduction}
A basic phenomenon underlying important quantum information processing tasks such as remote state preparation \cite{pati1999minimum}, quantum teleportation \cite{hu2020experimental}, device-independent (DI) certification of randomness \cite{pironio2010randm} and DI cryptography \cite{vazirani2019fully}, is the remote ``collapse'' or reduction of the quantum state of an entangled system $B$ by means of a local measurement on its partner system $A$. This phenomenon, which famously figures in the Einstein, Podolsky and Rosen (EPR) paradox \cite{einstein1935can}, lies at the heart of quantum nonlocality \cite{brunner2014bell}. As it happens, the basic nature of this remote measurement-disturbance of system $B$'s state-- in particular, the issue of whether the disturbance constitutes an objective change of $B$ in an operational sense, or merely a subjective update to the observer's knowledge of $B$--  has remained moot in quantum mechanics (QM).

This state of affairs is  part of the broader question whether the quantum state real, and is tied to the fact that the interpretation of QM is still not universally agreed upon \cite{leifer2014is,cabello2017interpretations}. Certain interpretations of QM are consistent with the idea that the quantum state is real \cite{cushing2013bohmian,beltrametti1995classical, penrose1996gravity, pusey2012on, bassi2013models, patel2017weak, ghirardi2020collapse}, while others that are inspired by the Copenhagen interpretation of QM \cite{faye2002copenhagen} are consistent with the idea that the quantum state only represents our knowledge about measurement outcomes or  underlying ontic variables \cite{fuchs2014introduction, peres2006quantum, brukner2003information, spekkens2007evidence, harrigan2010einstein}. Correspondingly, the wavefunction collapse takes on an ontic sense, or not, in these interpretations.

Bell's theorem \cite{brunner2014bell} itself can't help here. Its operational significance is not an indication of the reality of this remote disturbance, but rather a complementarity between signaling and unpredictability \cite{cavalcanti2012bell, hall2010complementary, kar2011complementary, aravinda2015complementarity, aravinda2016extending}. This leads, via the assumption of no-signaling, to the possibility of device-independent randomness generation \cite{herrero2017quantum} and key distribution \cite{vazirani2019fully}, but doesn't give any guidance on addressing the foundational  issue raised above.

These considerations prompt the question of whether a relevant concept of reality can be indicated that is independent of the physical interpretation or mathematical representation of QM. Here, we wish to show that this can be achieved based only on basic operational considerations about measurement-induced disturbance. Here it will be convenient to use the framework of generalized probability theories (GPTs) \cite{hardy2001quantum, barrett2007information, janotta2014generalized, chiribella2015entanglement}, which will make it clear that the concept of reality developed here is purely of operational origin and independent of the Hilbert space formalism. Our point of departure is the intuitive idea that if a measurement-disturbance can be the basis for communication, then it constitutes an objective fact and is thus real. Requiring only operational criteria and not extraneous paraphernalia such as hidden variables or parallel worlds, this approach promises to provide an \textit{intrinsic} interpretation of QM.

\color{black}

The remaining article is arranged as follows. Starting with a brief note on the concept of  ``operational'', we define the concept of operational reality of local measurement disturbance in Section \ref{sec:single}. This is then extended in Section \ref{sec:multi} to the case of  remote measurement-disturbance, which is shown  in Section \ref{sec:bell} to be related to Einstein-Podolsky-Rosen (EPR) steering. Building on this, a device-independent approach to indicate remote measurement-disturbance is discussed. Finally, we present our conclusions and related discussions in Section \ref{sec:conclu}. Here, we argue that the case for the operationally real nonlocality highlights the tension between quantum no-signaling and relativistic signal locality.

\section{Measurement disturbance and uncertainty in single systems \label{sec:single}} 
The operational formulation of a physical theory (such as quantum mechanics) consists in an abstract characterization of the theory as a GPT aka convex operational theory, i.e., in terms of rules governing preparation procedures, probabilities for measurement outcomes and reversible operations,  while avoiding concepts that cannot be accessed directly, such as the Hilbert space, complex global phase, etc. \cite{janotta2014generalized, barnum2011information}. In other words, the operational formulation corresponds to the basic ``syntax'' of a theory, devoid of the ``semantics'' pertaining to its physical interpretation or mathematical representation. 

An operational feature of a theory is one that can be defined as an element of its GPT formulation. For example, a state $\varphi$ is operationally understood as an equivalence class of preparation procedures. 
Given a bipartite state $\varphi_{AB}$, the marginal state of system $A$, denoted $\varphi_A^{\#}$, is the GPT analogue of the reduced density operator obtained via partial tracing over $B$. An entangled state in a GPT corresponds to a valid composite state that cannot be expressed as a convex combination of product states. 

In practice, we would require a detailed characterization of the subsystems in order to indicate a given joint state  in a GPT is entangled. A device-independent formulation represents a further level of abstraction where even such a characterization is not required, as the input-output statistics can self-test, typically based on the violation of Bell-type inequalities \cite{clauser1969proposed,brunner2014bell}.

A basic feature of QM, and indeed of any nonclassical GPT, is that the act of measurement can disturb-- i.e., randomly alter-- the measured state of the given system, even in the case of a \textit{pure} state. This phenomenon is ultimately due to the non-simpliciality of the relevant state space of the theory \cite{aravinda2017origin}.
\begin{defn}
Suppose measurement $x$ is performed on a quantum or GPT system $A$ prepared in state $\varphi_A$, producing outcome $a$, with probability $\mathfrak{p}(a|x,\varphi)$. Let the normalized post-measurement state be denoted $\varphi_A^{a|x}$. The state change 
\begin{equation}
\varphi_A \longrightarrow \varphi_A^{a|x},
\label{eq:mezdis}
\end{equation}
is called the measurement(-induced) disturbance of the initial state $\varphi_A$. \hfill $\blacklozenge$
\end{defn}
Here the measurement used, including the state update rule, may be considered as the GPT analogue of the L\"uders instrument in the context of quantum measurement \cite{busch1996quantum}.  As it happens, no experiment to date has succeeded in observing such disturbance of a quantum state, called a \textit{collapse} or \textit{state reduction} in this context.  The quantum measurement problem is concerned with the question of whether state reduction happens objectively. Revisiting this issue, we ask: Is the measurement disturbance represented in Eq. (\ref{eq:mezdis}) real? That is, does the measurement objectively alter system $A$, or is measurement disturbance only epistemic, i.e., one that only updates the observer's knowledge about a pre-existing property of $A$, without objectively altering $A$? 

Given a GPT, consider a communication protocol $\mathfrak{P}_1$ between Alice and Bob, wherein Alice transmits the non-classical system $A$ either in the state $\varphi_A$ or in the post-selected state $\varphi_{A}^{a|x}$. (She measures $x$ and discards states $\varphi_{A}^{a^\prime|x}$ for which $a^\prime \ne a$). Importantly, she doesn't send any supplementary classical information to Bob. By performing on $A$ a measurement $x^\prime$ incompatible with $x$ \cite{busch2013comparing}, Bob tries to determine which out of $\varphi_{A}$ and $\varphi_{A}^{a|x}$ Alice transmitted.  If he succeeds with a probability $p$ better than a random guess's (i.e., $p > \frac{1}{2})$,  then it is quite natural to infer that Alice's measurement objectively disturbed system $A$. Letting $\mathfrak{S}_{A\rightarrow A}$ denote the amount of information (in bits) that Alice can communicate on average to Bob using protocol $\mathfrak{P}_1$, we have the following intuitive, operational concept of the reality of Alice's measurement disturbance, applicable to any GPT:
\begin{defn}
Given state $\varphi_{A}$ and measurement $x$, if it is the case in protocol $\mathfrak{P}_1$ that for some outcome $a$
\begin{equation}
\mathfrak{S}_{A\rightarrow A} > 0,
\label{eq:saa}
\end{equation}
then the measurement disturbance Eq. (\ref{eq:mezdis}) constitutes an operationally real disturbance of $A$. $\hfill \blacklozenge$
\label{def:proximal}
\end{defn} 
Note that the concept of reality in Definition \ref{def:proximal} avoids any reference to  a ``hidden variable'' (HV) ontology. 
In a practical implementation of protocol $\mathfrak{P}_1$, it is not necessary to post-select on a specific outcome $a$. Let $\mathcal{M}_x(\varphi_A) \equiv \sum_a \mathfrak{p}(a|x) \varphi_{A}^{a|x}$ represent the non-selective state obtained by ignoring the outcome. Alice can use $\varphi_{A}$ and $\mathcal{M}_x(\varphi_{A})$ as symbols for communicating. If this works, then obviously $\Vert \varphi_{A} - \mathcal{M}_x(\varphi_{A})\Vert>0$, and by convexity, $\exists_a \Vert \varphi_A - \varphi^{a|x}_A\Vert > 0$, from which Eq. (\ref{eq:saa}) follows, and reality is inferred via Definition \ref{def:proximal}.

As an illustration in QM, suppose Alice chooses to perform or not to perform measurement $\sigma_X$ on qubit $A$ initialized in the state $\varphi_A = \cos^2(\theta/2)\ket{0}\bra0 + \sin^2(\theta/2)\ket{1}\bra1$, with $0 \le \theta \le \frac{\pi}{2}$ and $\theta \ne \pi/4$. She communicates with Bob by sending either the state $\varphi_A$ or $\mathcal{M}_{\sigma_X}(\varphi_{A}) := \frac{\mathbb{I}}{2}$. Since these two states are probabilistically distinguishable, it follows that $\mathfrak{S}_{A\rightarrow A} > 0$. Thus, we infer that Alice's measurement really disturbs the system in the operational sense. 

If $A$ is not an isolated single system, but entangled with another system $B$, then the above criterion may not be applicable. For example, suppose Alice and Bob share the bipartite state $
\ket{\Psi(\theta)}_{AB} =\cos(\theta)\ket{00}_{AB} +\sin(\theta)\ket{11}_{AB}, 
~~~\theta \in [0,\frac{\pi}{2}].
$
Alice measures $A$ in the basis $\sigma_Z$, and sends $A$ to Bob. As $\saa=0$ here, thus Definition \ref{def:proximal} is unable to indicate the reality of $A$'s measurement disturbance. 

In this situation, an entanglement-assisted version of protocol $\mathfrak{P}_1$ can be used, which we call $\mathfrak{P}_2$. Here, Alice and Bob share the entangled state $\varphi_{AB}$, with $A$ being initially with Alice and $B$ with Bob. As in $\mathfrak{P}_1$, Alice sends $A$ to Bob after measuring $x_0$ on it. Bob performs a joint measurement on $A$ and $B$ that is incompatible \cite{busch1984various,busch2013comparing} with $x_0$ (see Sec. \ref{sec:a1}). If he can (probabilistically) determine whether or not Alice measured $x_0$,  he infers that her measurement really disturbed $A$ in an operational sense.  We denote by $\mathfrak{S}_{A\rightarrow AB}$ the amount of information (in bits) about Alice's measurement choice that can be communicated to Bob in this way. In place of Definition \ref{def:proximal}, we have the following operational criterion (applicable to any sufficiently rich GPT):
\begin{defn}
	Given bipartite state $\varphi_{AB}$ and local measurement $x$ on $A$ in protocol $\mathfrak{P}_2$, if it is the case that
	\begin{equation}
	\saab > 0,
	\label{eq:saab}
	\end{equation}
	then the measurement disturbance Eq. (\ref{eq:mezdis}) constitutes an operationally real disturbance of $A$. $\hfill \blacklozenge$
	\label{def:proxijoint}
\end{defn} 

In the context of the above example with state $\ket{\Psi(\theta)}_{AB}$, the post-measurement \textit{bipartite} state with Bob is distinguishable from the initial bipartite state, i.e., $\rho_{AB}^{(Z)}\equiv \cos^2(\theta)\ket{00}_{AB}\bra{00} + \sin^2(\theta)\ket{11}_{AB}\bra{11} \ne \ket{\Psi(\theta)}_{AB}\bra{\Psi(\theta)}$, implying the satisfaction of Eq. (\ref{eq:saab}). Accordingly, Bob infers the operational reality of the measurement disturbance of $A$ per Definition \ref{def:proxijoint}. Analogous situations of disturbance can be pointed out for example in the case of Spekkens' toy theory \cite{spekkens2007evidence} and the box world \cite{oppenheim2010uncertainty}.




\section{Remote measurement-disturbance in a bipartite system \label{sec:multi}}
Suppose Alice and Bob share the quantum state $\ket{\Psi(\theta)}$. Alice measures $\sigma_X$ on qubit $A$ obtaining outcome $\ket{\pm} \equiv \frac{1}{\sqrt{2}}(\ket0 \pm \ket1)_A$. Correspondingly, Bob's particle $B$ collapses to the state $\ket{\theta^\pm} \equiv \cos(\theta)\ket0_B \pm \sin(\theta)\ket{1}_A$. Does this remote collapse constitute a real disturbance of system $B$?

More generally, in a GPT context, suppose Alice and Bob share an entangled state $\varphi_{AB}=\sum_{\lambda} p_\lambda \varphi_{AB}^\lambda$, $p(\lambda)$ being a probability distribution, where Bob's marginal (or, reduced) state is denoted $\varphi_{B}^\#$. Alice's measurement of $x$ on $A$ conditioned on her obtaining outcome $a$ leaves Bob's system in the unnormalized  (indicated by a tilde) state:
\begin{equation}
\tilde{\varphi}^{a|x}_{B} = \sum_\lambda p_\lambda \mathfrak{p}(a|x,\lambda)
\varphi_B^{a|x,\lambda}.
\label{eq:postmeasure}
\end{equation} 
Denote its normalized version by $\varphi_{B}^{a|x} \equiv \mathcal{N}\tilde{\varphi}^{a|x}_{B}$ where $\mathcal{N} \equiv [\sum_\lambda p(\lambda) \mathfrak{p}(a|x,\lambda)]^{-1}$. Does the state change 
\begin{equation}
\varphi_B^\# \longrightarrow \varphi_{B}^{a|x},
\label{eq:remmezdis}
\end{equation} 
which is the remote analogue of Eq. (\ref{eq:mezdis}), constitute a real disturbance of $B$ in an operational sense? That is, can one advance a purely operational argument in support of the claim that Alice's remote measurement-disturbance of system $B$ constitutes an objective change of $B$? 

An appeal here to the direct analogue of Definition \ref{def:proximal} is obviously ruled out by virtue of no-signaling, which entails that $\mathfrak{S}_{A\rightarrow B}=0$. 
Furthermore, in an instance where Eq. (\ref{eq:saab}) holds true, the signal would be attributed to the (local) disturbance of $A$ rather than to the (remote) disturbance of $B$. Thus, a simple adaptation of Definition \ref{def:proxijoint} is also ruled out. We now present an indirect, operational criterion to indicate the reality of a remote disturbance. 

We will require a specific feature of our ontological formalism, which is that it should satisfy a reasonable consistency principle in assigning reality to the measurement disturbances of multiple particles in a given measurement situation. In particular, given a pure joint state $\varphi_{AB}$, if the disturbances to $A$ and $B$ are identical (to each other) when either particle is measured, then the operational reality status of the two disturbances must also be identical.  This notion of consistency is natural since holding the disturbance of one of the particles to be real, but not that of the other  particle that is identically disturbed in the same measurement situation, would make this operationally inspired ontological system somewhat incoherent. Formally:
\begin{ftr}
	Suppose $\varphi_{AB}$ is a pure state such that the disturbances of $A$ and $B$ are identical (to each other) under measurement of $x$ on $A$. That is, the marginal states of $A$ and $B$ in $\varphi_{AB}$ are identical:	
\begin{equation}
\varphi_{A}^\# = \varphi_{B}^\#,
\label{eq:op1}
\end{equation}
and furthermore, their respective post-measurement states are also identical: 
\begin{equation}
\varphi_{A}^{a|x} = \varphi_{B}^{a|x}.
\label{eq:op2}
\end{equation}
Then, consistency requires that any attribution of operational reality to the  disturbances of $A$ and $B$ should be identical, i.e., either both disturbances are deemed operationally real, or both are deemed not.
$\hfill \blacklozenge$
\label{ftr:consistent}
\end{ftr}
Here we note that in QM, any pure bipartite state has the symmetric property of Eq. (\ref{eq:op1}). Therefore,  to fulfill the conditions of Feature \ref{ftr:consistent}, it suffices to choose a suitable measurement satisfying the symmetry property Eq. (\ref{eq:op2}). 
Feature \ref{ftr:consistent} is a means for (indirectly) addressing the question of the reality of $B$'s remote measurement-disturbance Eq. (\ref{eq:remmezdis}), given the constraint of no-signaling. This argument may be formalized as follows: 
\begin{theorem}
Given pure state $\varphi_{AB}$, and Alice's measurement $x_0$ on system $A$, if the (local) measurement-disturbance of $A$ is: (a) operationally real per the criterion of Definition \ref{def:proxijoint}, and:  (b) is identical to the (remote) measurement-disturbance of system $B$, then the latter disturbance is operationally real. 
	\label{def:distal}
\end{theorem}

Intuitively, the idea here is that if the correlations between $A$ and $B$ are local, then Alice's measurement should operationally disturb only $A$, but not $B$. In particular, non-trivial disturbances of $A$ and $B$ couldn't be identical. But Theorem \ref{def:distal} gives conditions under which $A$ is operationally disturbed and yet this identicality holds, thereby negating the locality supposition. The proof below of Theorem \ref{def:distal} is essentially a formal elaboration of this idea.

\begin{proof}
Suppose that locality holds on the operational level, meaning that Alice's measurement on $A$ doesn't disturb the state of system $B$ in an operational sense. More precisely, there exists a pre-existing ensemble of states on Bob's side $\chi \equiv \{q(\mu), \phi_B^\mu\}$, such that the post-measurement state of $B$ can be expressed as 
$
\tilde{\varphi}^{a|x}_{B} = \sum_\mu q(\mu) \mathfrak{q}(a|x,\mu) \phi_B^{\mu}, 
$
where $q(\mu)$ and $\mathfrak{q}(a|x,\mu)$ are probability distributions. It follows in view of the symmetry condition (b), specifically requirement Eq. (\ref{eq:op2}), that the joint state of the composite system $AB$ conditioned on Alice's measurement, is given by 
$$\tilde{\varphi}^{a|x}_{AB} = \sum_\mu q(\mu) \mathfrak{q}(a|x,\mu) \varphi_{A}^\mu\varphi_B^{\mu},$$
such that $\varphi_A^{\mu}=\varphi_B^{\mu}$. This implies that, from Bob's perspective the non-selective joint state is given by 
$$
\sum_a \tilde{\varphi}^{a|x}_{AB} = \sum_\mu q(\mu) \varphi_A^{\mu}\varphi_B^{\mu},
$$ 
which is independent of $x$, entailing that Alice's choice $x$ could not be deduced by Bob via his joint measurement of $AB$, contradicting condition (a). Therefore, given the satisfaction of condition (a), it follows that particle $B$ is necessarily disturbed from afar by Alice's measurement in the sense that no such local ensemble $\chi$ exists for Bob. In this case, the general expression Eq. (\ref{eq:postmeasure}) should hold, where the specific outcome state $\tilde{\varphi}^{a|x}_{AB}$ depends on $x$.

To show that this remote disturbance of $B$ is operationally real, we note that by condition (b), the disturbance at $B$ is identical to that at $A$. By Feature \ref{ftr:consistent}, the same reality status should be assigned to both these disturbances. Therefore, the operational reality of $B$'s disturbance follows, given that of $A$'s disturbance by virtue of fulfillment of condition (a).
\end{proof}

A state $\varphi_{AB}$ that admits such a remote measurement-disturbance in the above sense may conveniently be called nonlocal in an \textit{operationally real} (OR) sense. The state $\ket{\Psi(\theta)}$ is evidently OR nonlocal according to Theorem \ref{def:distal}, with measurement $x_0 \equiv \sigma_Z$. 
 
In the standard Bell test and EPR steering scenarios, the particles $A$ and $B$ are measured in geographically separated stations. By contrast, in the present scenario, Alice sends $A$ to Bob such that he can subsequently perform a joint measurement on the composite system $AB$. Specifically, condition (a) above corresponds to the violation of a type of no-signaling-in-time (NSIT) by particle $A$. In the context of temporal correlations, NSIT is a statistical characterization of non-invasive measurability \cite{kofler2013condition, clements2016fine}. Given that both particles are available at the same place for the second measurement in the above scenario, it would be desirable to elucidate the operational sense in which Theorem \ref{def:distal} establishes the nonlocality of the correlations between $A$ and $B$. 

To this end, we may consider an alternative, equivalent protocol that makes this nonlocality explicit. We now make the important if somewhat straightforward observation that conditions (a) and (b) can be guaranteed through local operations and classical communication (LOCC), without requiring Alice to transmit particle $A$ to Bob for his joint measurement. We note that the initial state $\varphi_{AB}$ must be entangled in order for Alice to be able to re-prepare $B$'s state and thereby disturb it. Alice's measurement disentangles $\varphi_{AB}$ into a product state $\hat{\varphi}_A^{a|x} \otimes \tilde{\varphi}_B^{a|x}$, where $\hat{\varphi}_A^{a|x}$ denotes the normalized outcome obtained by Alice on $A$, and $\tilde{\varphi}_B^{a|x}$ is given by Eq. (\ref{eq:postmeasure}). Both the initial state $\varphi_{AB}$ and the final product state can be ascertained by LOCC in a GPT with the local tomographic property. Thus, the symmetry condition (b) can be directly checked. Moreover, noting that $\sum_a \hat{\varphi}_A^{a|x} \otimes \tilde{\varphi}_B^{a|x} \ne \varphi_{AB}$, it follows that condition (a) can also be checked by LOCC. 

The above considerations show that the conditions of Theorem \ref{def:distal}, and thus the reality of the remote disturbance, can be verified using only LOCC and without any quantum or nonclassical communication from Alice to Bob, thereby bringing the scenario of the OR nonlocality closer to that of Bell nonlocality or EPR steering. As shown below, this allows OR nonlocality to be deduced by a steering-like statistical inequality. Later, this will be shown to correspond to a one-sided device-independent (DI) characterization of OR nonlocality.

Mixed states present a kind of ``operational preparation-contextuality'', somewhat reminiscent of HV-ontological preparation contextuality \cite{spekkens2005contextuality}.  This is related to the fact that a nonclassical GPT,  being characterized in general by a non-simplicial state space, admits multiple pure-state decompositions of a given mixed state $\varphi_{AB}$, whereas the concept of reality will arguably depend on the actual decomposition that is the case.  Accordingly, with regard to indicating OR nonlocality, we have the following criterion.
\begin{defn}
\label{def:mix}
Given state $\varphi_{AB}$ prepared by a known probabilistic procedure of mixing an OR nonlocal component and a separable component, with probabilities $p$ and $1-p$, respectively, the remote measurement-disturbance of $B$ in state $\varphi_{AB}$ is  OR nonlocal with probability $p$.
    $\hfill \blacklozenge$
\end{defn}

For example, suppose Charlie prepares state $\ket{\Psi(\theta)}$ with probability $f$ and the four computational basis states $\ket{00}, \ket{01}, \ket{10}$ and $\ket{11}$ with equal probability $(1-f)/4$. He sends the first particle to Alice and the second to Bob. Accordingly, Alice and Bob share the Werner-like state 
$
W(f,\theta) \equiv f\ketbra{\Psi(\theta)} + (1-f)\frac{\mathbb{I}_4}{4},
$
with mixing parameter $f \in [0,1]$. Then Charlie can assert that the remote measurement-disturbance of Bob's particle is operationally real with probability $f$, with the choice $x_0 \equiv \sigma_Z$. 

As well, the state $W(f,\theta)$ can be prepared by Charlie mixing the state $\ket{\Psi(\theta+\frac{\pi}{2})}$ with probability $\frac{1-f}{2}$ and $\ket{\Psi(\theta)}$ with probability $\frac{1+f}{2}$. Since both $\ket{\Psi(\theta)}$ and $\ket{\Psi(\theta +\frac{\pi}{2})}$ are nonlocal in an operationally real sense, Charlie asserts that given the state $W(f,\theta)$ under this preparation, the remote measurement-disturbance of Bob's particle is operationally real with probability 1.  Setting $f \equiv 0$ here, we find that even a maximally mixed state can be potentially OR nonlocal to the maximum extent. This surprising observation essentially has to do with the idea that the underlying reality of disturbance should be unaffected by the observer's state of knowledge.

If the preparation information of a given mixed state $\varphi_{AB}$ is unavailable, then its remote measurement-disturbance under measurement $x$ is said to be probabilistically OR nonlocal if $\varphi_{AB}$ contains a non-vanishing OR nonlocal component under any pure-state decomposition.  Here, a direct application of conditions (a) and (b) of Theorem \ref{def:distal} can be misleading. For example, the separable state $\half(\ket{\theta^+}\bra{\theta^+} \otimes \ket{0}\bra{0} + \ket{\theta^-}\bra{\theta^-} \otimes \ket{1}\bra{1})$, when $x \equiv \sigma_Z$, satisfies (a) and also, with  high probability, condition (b), by choosing sufficiently small $\theta$. Yet, it is evident that there can be no remote measurement-disturbance of $B$ for any separable state.

Therefore, to obtain a lower bound on OR nonlocality a different approach is required for mixed states. Here a key observation, noted above in the proof of Theorem \ref{def:distal}, is that Alice steers Bob's state. Given a steerable state $\rho$, the required lower bound is associated with the pure-state decomposition that gives the lowest fraction of OR nonlocality. In the above example of the state $W(f,\theta)$, the first decomposition provides the minimal decomposition. Thus if $f$ is sufficiently large to guarantee steerability, then the mixture is OR nonlocal with probability at least $f$. 

To witness steering in an arbitrary GPT, we use the uncertainty principle. This refers to the feature whereby two or more observables cannot simultaneously assume exact values \cite{oppenheim2010uncertainty}. Let $\mathfrak{p}(b|y)$ represent the probability of outcome $b$ upon measurement of $y$ on a given operational state $\varphi$. Define $\mathcal{P}(y) \equiv \max_b  \mathfrak{p}(b|y)$. Given measurements $y_0, y_1$ and $y_2$, an uncertainty relation exists if for any state $\varphi$
\begin{equation}
\mathcal{P}(y_0) + \mathcal{P}(y_1) + \mathcal{P}(y_2) \le \upsilon
\label{eq:unc0}
\end{equation}
such that $\upsilon < 3$. For classical theory, $\upsilon=3$ for any triple of sharp measurements.
Here it is assumed that the measurements $y_j$ aren't ``trivial'', such as one that produces a fixed outcome for any measured state $\varphi$. In this case, of course there is trivially no uncertainty.

Consider protocol $\mathfrak{P}_3$, where Alice and Bob share the state $\varphi_{AB}$. Alice performs measurement $x_j$ on the particle $A$ and predicts the outcome for Bob, who performs the corresponding measurement $y_j$ on $B$. We consider the conditional version of Eq. (\ref{eq:unc0}), namely:
\begin{equation}
\mathcal{P}(y_0|x_0) + \mathcal{P}(y_1|x_1) + \mathcal{P}(y_2|x_2)
~\le~ \upsilon,
\label{eq:eta}
\end{equation}
where $\mathcal{P}(y_j|x_j)$ represents Bob's certainty in $y_j$ measured on $B$, conditioned on Alice's measurement of $x_j$ on $A$ and knowing the outcome. The violation of Eq. (\ref{eq:eta}) certifies that Alice's measurement genuinely re-prepares-- and thereby disturbs-- the state of $B$. As such, Eq. (\ref{eq:eta}) represents an EPR (spatial) steering inequality. A necessary condition here is that the measurement pairs $x_0, x_1$ and $x_3$ must be pairwise incompatible for its violation (Appendix, Part II). 

Given a pure state $\varphi_{AB}$ in a GPT and a pair of measurements $(x_0, y_0)$, if measuring $x_0$ on $A$ and $y_0$ on $B$ produce identical results, i.e., Eq. (\ref{eq:op2}) holds, then the pair $(x_0, y_0)$ is said to be commensurate for $\varphi_{AB}$. For a mixed state prepared by combining such $\varphi_{AB}$ (with probability $p$) and other states, the pair $(x_0, y_0)$ is said to be commensurate with probability (at least) $p$. For example, the measurements $x_0 = y_0 := \sigma_Z$ is commensurate for the state $\ket{\Psi(\theta)}$. Consider a mixed states such as $W(f,\theta)$, that are known to be noisy versions of pure states that have identical marginal states (i.e., satisfy Eq. (\ref{eq:op1})) and either admit a pair of commensurate measurements or can be brought to that form by application of suitable local reversible operations (i.e., satisfy Eq. (\ref{eq:op2})). 
Promised a mixed state with this commensurate measurement property, the following result shows that the degree of violation of Eq. (\ref{eq:eta}) can be used to lower-bound the degree of OR nonlocality.

\begin{theorem}
	Given 
	the violation of the inequality Eq. (\ref{eq:eta}) with observed correlation $\upsilon^\ast>\upsilon$, the remote measurement-disturbance of $B$ under measurement $x_0$ on $A$ is operationally real with probability at least $\frac{\upsilon^\ast-3/2}{\upsilon_{\max}-3/2}$, where $\upsilon_{\max}$ is the largest violation ($\le 3$) allowed in the given GPT.
	\label{thm:main}
\end{theorem}
\begin{proof}
In Eq. (\ref{eq:postmeasure}), if the remote repreparation of the state of $B$ can be explained by a pre-existing, hidden state ensemble of $B$-states $\{q(\mu), \phi_B^\mu\}$,  then conditioning on $x$ won't provide information to beat the uncertainty relation Eq. (\ref{eq:unc0}) for system $B$. Therefore, the EPR steering inequality Eq. (\ref{eq:eta}) will hold. It follows that under its violation, the measurement disturbance of $A$ disturbs the system $B$.  

Let the minimal pure fraction leading to the violation in all possible pure state decompositions of the given state $\varphi_{AB}$ be $f$. Then $f \ge f_{\min}$, where $f_{\min}\upsilon_{\max} + (1-f_{\min})\frac{3}{2} = \upsilon^{\ast}$, and $\upsilon_{\max}$ and $\frac{3}{2}$ are, respectively, the theory-dependent maximum and the algebraic minimum value attainable by $\upsilon$ in Eq. (\ref{eq:eta}). Solving the equation, we find $f_{\min} = \frac{\upsilon^{\ast}-3/2}{\upsilon_{\max}-3/2}$. By assumption, the mixed state is a noisy version of a pure state with identical marginal states and the commensurate measurement associated with the pair $(x_0, y_0)$. Therefore, with probability at least $f_{\min}$, the conditions of Theorem \ref{def:distal} must be satisfied.
\end{proof}

To clarify Theorem \ref{thm:main}, we note that the violation of the inequality Eq. (\ref{eq:eta}), as discussed, entails the remote preparation of $B$'s state (Eq. (\ref{eq:postmeasure})) and hence the disentanglement of the initial state $\varphi_{AB}$. This in turn implies the verification of the reality condition (a) in the equivalent LOCC scenario. However, quantifying how closely condition (b) is supported by the experimental conditional probabilities $\mathcal{P}(y_j|x_k)$ would require certain theory-dependent assumptions. Furthermore, experimentally estimating the violation of inequality Eq. (\ref{eq:eta}) in terms the theory-specific quantity $\upsilon_{\max}$ presupposes that the states of $B$ are well characterized. On the other hand, the probability $\mathfrak{p}(a|x,\lambda)$ in Eq. (\ref{eq:postmeasure}) can be arbitrary. This situation corresponds to \textit{one-sided device-independence}, which is of practical importance in cryptography \cite{branciard2012one}. In Section \ref{sec:bell}, we show how the above result can be strengthened to a fully device-independent characterization of OR nonlocality. 

As a specific realization of Eq. (\ref{eq:eta}), we consider $\mathcal{P}(y_0|x_0) + \mathcal{P}(y_1| x_1, a_1) + \mathcal{P}(y_2| x_1, \overline{a}_1) \le \upsilon$, where Alice has only two measurement choices: $x_0$ and $x_1$, whilst Bob has three. If Alice measures $x_0$, then Bob measures $y_0$. If her measurement is $x_1$, then corresponding to her outcome $a_1$ (resp., $\overline{a}_1$), Bob measures $y_1$ (resp., $y_2)$. 	For an application to the quantum context,  we set $x_0=y_0 :=\sigma_Z$, $x_1 :=\sigma_X$, $y_1 := \sin(2\theta)\sigma_X + \cos(2\theta)\sigma_Z$ and $y_2 := \sin(2\theta)\sigma_X  - \cos(2\theta)\sigma_Z$. For these settings of Bob, the single system uncertainty bound is given by $\upsilon  = \frac{5}{2}$. 

 The state $\ket{\Psi(\theta)}$ under the above settings entails a violation of Eq. (\ref{eq:eta}) up to its algebraic maximum of 3 for any $\theta$ in the above range.  Moreover, the margin of maximal violation over the local bound $\upsilon$ is $\frac{3}{\upsilon} = \frac{6}{5}$, for the optimal choice $\theta=\pi/6$ here. This can be shown to be larger than the optimal margin of $\frac{2}{\upsilon} = \frac{2\sqrt{2}}{\sqrt{2}+1}\approx 1.17$ in the case of the analogous two-term steering inequality (Appendix). Thus, inequality Eq. (\ref{eq:eta}) is  suitable for the non-maximally entangled state $\ket{\Psi(\theta)}$.

\section{Relation to Bell nonlocality \label{sec:bell}}
Bell-nonlocality is a stronger condition than steering. 
With measurement settings as above, the bound $\upsilon :=\frac{5}{2}$ in Eq. (\ref{eq:eta}).
If $f >  \frac{2}{3}$, Werner-like states $W(f,\frac{\pi}{6})$ violate the inequality.
 Using the two-qubit nonlocality criterion \cite{horodecki1995violating}, we find that the state $W(f,\frac{\pi}{6})$ is Bell nonlocal for $f > \frac{2}{\sqrt{7}}$. Thus, in the range $f \in [\frac{2}{3},\frac{2}{\sqrt{7}}]$, the state $W(f,\frac{\pi}{6})$ is Bell local but nonlocal in an operationally real way. It is important to stress that the existence of such Bell-local states doesn't undermine the operational reality of the remote measurement-disturbance precisely because the hidden variables of the Bell-local theory are not part of the operational theory of our interest. Indeed, the basic premise here that we decide the reality of a disturbance by operational considerations alone, without recourse to hidden-variable ontology. 
Another point is that given correlation $P(a,b|x,y)$ of OR-nonlocal but Bell-local correlations in a GPT, the classical dimension $|\lambda|$ of shared randomness in a local model will be larger than the dimension $\mathfrak{d}$ of the correlated systems of the GPT, i.e., such correlations will be \textit{superlocal} \cite{donohue2015identifying,jebaratnam2017nonclassicality}. It is reasonable to assume that the GPT encompasses classical theory and so any local correlation requiring shared randomness $|\lambda| \le \mathfrak{d}$ can be produced by local measurements on separable states of the GPT. Thus, OR nonlocality entails superlocality, i.e., $|\lambda| > \mathfrak{d}$.

Finally, let us point out that the above characterization of OR nonlocality in the LOCC scenario naturally leads to a device-independent (DI) characterization thereof. The idea is to obtain a sufficient condition for OR nonlocality without reference to  theory-dependent parameters such as $\upsilon_{\max}$. We suppose that in protocol $\mathfrak{P}_3$, Alice and Bob perform their measurements simultaneously on the pre-shared state $\varphi_{AB}$.  Importantly, there is no classical communication from her to Bob or vice versa until after their measurements. Further, we suppose that from a subset of the resulting conditional probabilities $P(a,b|x,y)$, Alice and Bob construct a Bell-type inequality, e.g.,
\begin{equation}
P_{11}^= + P_{12}^= +P_{21}^= + P_{22}^{\ne} \le 3,
\label{eq:bell}
\end{equation}
where $P_{jk}^= \equiv P(a_j=b_k|x_j, y_k)$ and $P_{jk}^{\neq} \equiv P(a_j \neq b_k|x_j, y_k)$. In an arbitrary GPT, under maximal violation of this inequality, the quantity in the l.h.s can go up to the algebraic maximum of 4, while in QM it is $2 + \sqrt{2}\approx 3.414$. 

By design, if Alice's conditional probability  has the local-realist form $P(b|a,x,y) = \sum_\lambda p_\lambda \mathfrak{p}(a|x,\lambda) \mathfrak{p}_B(b|y,\lambda)$, i.e., Bob's measurement merely reveals a pre-existing value of $y$ on $B$, then the satisfaction of the inequality Eq. (\ref{eq:bell}) follows. 
Accordingly, the violation of the inequality implies the absence of such a value and can thereby serve as the basis to certify the entanglement in state $\varphi_{AB}$ in a DI manner, i.e., irrespective of the details of the operational theory governing the subsystems $A$ and $B$. Alice and Bob may perform a follow-up measurement on their respective particles to confirm the absence of entanglement in the post-measured particles, which trivially requires no theory-specific assumptions. 
By thus verifying the disentangling action of their measurements, they obtain a DI checking of the reality condition (a) of Theorem \ref{def:distal}.

As regards condition (b), it turns out that in the context of device independence, it can be relaxed, by not requiring the disturbances of particles $A$ and $B$ to be identical. As the violation of Eq. (\ref{eq:bell}) precludes the possibility of a pre-existing value of $y$ on $B$, Bob's conditional state here is in a sense created by Alice's act of measurement. (To improve the semantics, we let Bob's measurement to happen slightly later in their common reference frame, but in such a way that their measurements are spacelike separated.) Ontological consistency of the formalism then requires that the remote measurement-disturbance of $B$ be OR, given that the measurement-disturbance of $A$ is OR under a violation of inequality Eq. (\ref{eq:bell}). For otherwise, we would have the incoherent situation that $A$'s measurement disturbance is real, and yet not so the remote state-preparation that it is certified (by the Bell inequality) to have produced. This observation provides a natural extension to the consistency requirement of Feature \ref{ftr:consistent} in the DI scenario. 

We thus have the following strengthening of Theorem	\ref{thm:main}:
Given the violation of the inequality Eq. (\ref{eq:bell}), the remote measurement-disturbance of $B$ under measurement $x$ on $A$ is operationally real with a non-zero 
probability. 

\section{Discussion and conclusions \label{sec:conclu}}

The ontological question of whether measurement-induced disturbance is real in QM or other nonclassical operational theories is addressed in the present work by employing-- ironically-- only operational considerations. As such, it  can be construed as providing a different response to the EPR  paradox \cite{einstein1935can, wiseman2013quantum}, than both Bohr's \cite{bohr1935can} and Bell's \cite{bell1964on} responses. 

EPR asserted essentially that quantum spatial steering entailed a ``spooky action-at-a-distance'' \cite[Article 16]{bell2004speakable}, which they hoped could be banished in a more complete version of QM. Employing a different criterion of reality than EPR, the present work argued that the action-at-a-distance is an unavoidable feature of operational QM itself, and not just of QM's ontological completion (as follows from Bell's theorem).

The operationally reality of remote measurement-disturbance  sheds new light on the tension \cite{gisin2009quantum} between quantum nonlocality and special relativity. In particular, it highlights that quantum no-signaling is distinct from relativistic signal locality. The former is a consequence of the tensor product structure of the state space, with no  association to light speed, whereas the latter is essentially a prohibition on superluminal transmission of information arising from the Lorentz invariance of the light cone. Evidently, these two no-go conditions belong to two distinct frameworks, respectively. In point of fact, quantum nonlocality  is non-signaling even in non-relativistic QM. 


In an instance of OR nonlocality, as far as Alice can say, when she measures $x$ and obtains outcome $a$, particle $B$ is instantaneously left in the state $\tilde{\varphi}_B^{a|x}$, given by Eq. (\ref{eq:postmeasure}). The remote measurement-disturbance of $B$, given by Eq. (\ref{eq:remmezdis}), thus represents a spacelike influence linking the event of $A$'s measurement and the event of re-preparation of $B$'s state. Per Theorem \ref{thm:main}, this influence is nevertheless real in an operational sense, and thus arguably imposes an intrinsic time-ordering on the two events. This observation underscores a further aspect of the distinction between quantum no-signaling and signal locality in special relativity. Whereas the latter is conceivably a not unanticipated speed-limit on information propagation, by contrast the former seems surprising in light of OR nonlocality, and prompts the question-- especially relevant in the context of reconstructing QM from operational or information theoretic principles \cite{clifton2003characterizing, chiribella2011informational, dakic2011quantum}-- of why an OR nonlocal theory, such as QM, is non-signaling. Indeed, we understand the complementaristic role played here by the randomness of measurement outcomes in suppressing the signaling (cf. \cite{aravinda2015complementarity, kar2011complementary, hall2010complementary}), but this only answers the ``how'', rather than the ``why'' aspect of the question.  

Here we offer a simple and brief answer: that no-signaling is a consequence of the natural requirement of consistency between the properties of single systems and those of reduced systems. Suppose no-signaling could be violated in a GPT. Then, the inequality Eq. (\ref{eq:eta}) could be violated without Alice's classical communication of the outcome of her measurement $x_j$. But if so, then it would essentially mean that violation of the local uncertainty principle Eq. (\ref{eq:unc0}) at Bob's end, contradicting this property in the context of single systems.

\acknowledgments
The author thanks S. Aravinda for discussions. He acknowledges the support of Dept. of Science and Technology (DST), India, Grants No. MTR/2019/001516 and also the Interdisciplinary Cyber Physical Systems (ICPS) programme of the DST, India, Grant No. DST/ICPS/QuST/Theme-1/2019/14.

\bibliography{qvanta}


\appendix

\section{Measurement incompatibility and disturbance \label{sec:a1}}

Suppose measuring $x_0$ doesn't disturb the joint state $\varphi_{AB}$, i.e., $\mathcal{M}_{x_0 \otimes u}(\varphi_{AB}) = \varphi_{AB}$ (where $u$ represents an identity operation on $B$). Let $x_J$ be the joint measurement on $A$ and $B$ that is used to check for disturbance of $A$. This would entail that $x_0 \otimes u$ and $x_J$ are compatible, because we can construct their joint measurement, given by $J(a_0,a_J|x_0,x_J,\varphi_{AB}) \equiv \mathfrak{p}(a_0| x_0,\varphi_{AB}) \mathfrak{p}(a_J|x_J,\varphi_{AB})$ for this state, simply by first measuring $x_0$ and then $x_J$, and noting the respective outcome probabilities. Thus, the measurement of $x_0$ must (globally) disturb $\varphi_{AB}$. This disturbance can in principle be detected by joint measurement $x_J$ on systems $A$ and $B$, whereby $\saab > 0$. 
	
%

\section{Incompatibility and steering \label{sec:a2}}

The violaton of Eq. (\ref{eq:eta}) implies that the measurements $x_0, x_1$ and $x_2$ are pairwise incompatible. To show this, for simplicity consider the two-term variant of the above inequality:
\begin{equation}
\mathcal{P}(y_0|x_0) + \mathcal{P}(y_1|x_1) ~\le~ \upsilon_2,
\label{eq:eta2}
\end{equation}
where $\upsilon_2$ is the local uncertainty bound for two measurements. Suppose $x_0$ and $x_1$ are jointly measurable  \cite{busch1984various,busch2013comparing}.  In place of Eq. (\ref{eq:postmeasure}), we would have $\tilde{\varphi}^{a_0,a_1|x_0,x_1}_{B}  \equiv  \sum_\lambda p(\lambda)  \mathfrak{p}(a_0,a_1|x_0,x_1,\lambda)    
 \varphi_B^{a_0,a_1|x_0,x_1,\lambda}$, where the conditional probability $\mathfrak{p}(a_j|x_j,  \lambda)$ for either measurement should be  derivable as the marginal statistics of a ``master measurement'': $\mathfrak{p}(a_j|x_j,\lambda) =  \sum_{a_{\overline{j}}}
\mathfrak{p}(a_0, a_1|x_0, x_1, \lambda)$ where $\overline{j} \equiv j + 1 {\rm ~mod~} 2$. But this means that the states $\varphi_B^{a_0,a_1|x_0,x_1,\lambda}$ constitute hidden states to reproduce the result of the two measurements. Specifically, to implement the measurement of $x_0$, one marginalizes over $a_1$: $\tilde{\varphi}^{a_0|x_0}_B \equiv \sum_{a_1} \tilde{\varphi} ^{a_0,a_1|x_0,x_1}_{B}$, and vice versa.

\end{document}